\documentclass[11 pt]{article}
\usepackage{indentfirst,mathrsfs}
\usepackage{amsfonts}
\usepackage{amsmath,amsthm,amssymb}
\usepackage[colorlinks,linkcolor=black,
            pdftitle={title},
              pdfauthor={author},
              pdfkeywords={}]{hyperref}

\newtheorem{theorem}{Theorem}

\newtheorem{construction}{Construction}
\newtheorem{remark}{Remark}

\newtheorem{lemma}{Lemma}

\def\F{{\mathbb F}}
\newcommand{\0}{\mathbf{0}}

\begin{document}

\begin{center}
\Large{A general secondary construction of Boolean functions including the indirect sum and its generalizations}
\end{center}

\begin{center}
Claude Carlet\footnotemark[1] and Deng Tang\footnotemark[2]

\footnotetext[1]{Universities of Paris 8, France and Bergen, Norway. Email: claude.carlet@gmail.com}

\footnotetext[2]{Shanghai Jiao Tong University, China. Email: dtang@foxmail.com}
\end{center}

\begin{abstract}We {\color{black}study} a secondary construction of Boolean functions, which generalizes the direct sum and the indirect sum. 
We detail how these two classic secondary constructions are particular cases of this more general one, as well as two known generalizations of the indirect sum. This unifies the known secondary constructions of Boolean functions. We study very precisely the Walsh transform of the constructed functions. This leads us to an interesting observation on the Walsh transforms $W_g,W_{g'},W_{g''}$, and $W_{g\oplus g'\oplus g''}$ when $g,g',g''$ are Boolean functions such that $(g\oplus g')(g\oplus g'')$ equals the zero function.

\end{abstract}
\section{Introduction}
Constructions of Boolean functions satisfying some desired features (such as being highly nonlinear, and/or plateaued, and/or resilient, and/or algebraic immune see e.g. \cite{Book}) are necessary for providing Boolean functions that are applicable to coding theory, combinatorics, or cryptography, in information protection and discrete mathematics. Primary constructions (which build functions from scratch, such as the well-known Maiorana-McFarland construction \cite{Dil74,CC-CCCS91}) are often not sufficient for obtaining enough functions having all the desired features. Secondary constructions (building new functions from so-called initial functions having the desired property) are then also needed. The two main known secondary constructions are the direct sum construction (introduced in \cite{Dil74}) $f(x,y)=g(x)\oplus h(y)$ (where $x\in \F_2^r$ and $y\in \F_2^{n-r}$ and $(x,y)\in \F_2^n$ is the concatenation of $x$ and $y$) and its generalization, the indirect sum construction (introduced in  \cite{carlet2004resbent}): $f(x,y)=g(x)\oplus h(y)\oplus (g(x)\oplus g'(x))(h(y)\oplus h'(y))$ (where $x\in \F_2^r$ and $y\in \F_2^{n-r}$ as well). In \cite{ZCHZ}, the indirect sum is used with  $f_1$ and $f_2$ (resp. $g_1$ and $g_2$) being the restrictions of a bent function $f$ (resp. $g$) to two hyperplanes, complementary of each other.\\
The indirect sum has been itself generalized in several ways. In \cite{CZH2012} are introduced the constructions:
    \begin{equation}\label{gen1}f(x,y)=g(x)\oplus  h(y)\oplus  (g\oplus  g')(x)(h\oplus  h')(y)
    \oplus  (g'\oplus  g'')(x)(h'\oplus  h'')(y),\end{equation}  $$f(x,y)=g(x)\oplus  h(y)\oplus  (g\oplus  g')(x)(h\oplus  h')(y)
    \oplus (g'\oplus  g'')(x)(h'\oplus  h'')(y)\oplus $$\begin{equation}\label{gen2}  (g''\oplus  g''')(x)(h''\oplus  h''')(y).\end{equation}
   {\color{black} In the present paper, we study a secondary construction, which generalizes all these ones (and gives then a simpler view on them).\\

The study of the properties, needed for a Boolean function to be used for instance in a stream cipher, can be made, in almost all cases, through the study of its Walsh transform. The Walsh transform of an $n$-variable Boolean function $f:\F_2^n\mapsto \F_2$ is the integer-valued function $a\mapsto W_f(a)=\sum_{x\in \F_2^n}(-1)^{f(x)\oplus a\cdot x}$, where ``$\cdot$" is some inner product in $\F_2^n$, for instance the usual inner product $a\cdot x=\bigoplus a_ix_i$ (where $a=(a_1,\dots ,a_n)$ and $x=(x_1,\dots ,x_n)$).   The Walsh transforms of direct and indirect sums and of the generalization (\ref{gen1}) can be nicely expressed by means of the Walsh transforms of the initial functions $g,h,\dots$ (and this is why they easily provide bent functions from bent functions, plateaued functions from plateaued functions and resilient functions from resilient functions {\color{black}- see the definitions of these notions  in Section \ref{preliminarie}, see also the survey} \cite{Book}): we have, in the case of direct sum $W_f(a',a^{''})=W_g(a')W_h(a^{''})$, and in the case of indirect sum $W_{f}(a',a^{''})=
\frac{1}{2}W_{g}(a') {\Big (}W_{h}(a^{''})+W_{h'}(a^{''}){\Big )}+
\frac{1}{2}W_{g'}(a') {\Big (}W_{h}(a^{''})-W_{h'}(a^{''}){\Big )}$. The Walsh transform was not calculated in \cite{CZH2012} for Constructions (\ref{gen1}) and (\ref{gen2}), since the aim of this paper was only to build bent functions and characterizing the bentness of Functions (\ref{gen1}) and (\ref{gen2}) was possible by using a theorem from \cite{CC-C-constr}.  But the Walsh transform of (\ref{gen1}) was calculated in \cite{ZCHC} (in a rather complex way); the derived formula is:
$$W_f(a',a^{''})= \frac{1}{4}W_{h}(a^{''})\left[ W_{g}(a')+W_{g'}(a')\right.
  \left.+W_{g^{''}}(a')+W_{g\oplus g'\oplus g^{''}}(a')\right]$$$$+\frac{1}{4}W_{h'}(a^{''})\left[ W_{g}(a')-W_{g'}(a')\right.\left.-W_{g^{''}}(a')+W_{g\oplus g'\oplus g^{''}}(a')\right]$$$$+\frac{1}{4}W_{h^{''}}(a^{''})\left[ W_{g}(a')-W_{g'}(a')\right.\left.+W_{g^{''}}(a')-W_{g\oplus g'\oplus g^{''}}(a')\right]$$\begin{equation}\label{walshant}+\frac{1}{4}W_{h\oplus h'\oplus h^{''}}(a^{''})\left[ W_{g}(a')\!+\!W_{g'}(a')\!-\!W_{g^{''}}(a')\!-\!W_{g\oplus g'\oplus g^{''}}(a')\right].\end{equation}
 {\color{black}We study the Walsh transform of the functions generated by the generalized secondary construction. We check that the Walsh transforms of the direct and indirect sums, which correspond to the simplest particular cases of our construction, coincide with what gives our general formula. We study other (more complex) particular cases and we deduce new bent functions.}

\section{Preliminaries}\label{preliminarie}

For any positive integer $n$, we denote by $\F_2^n$ the vector space of $n$-tuples over the
finite field $\F_2=\{0,1\}$, by $\mathbf{0}_n=(0,\ldots,0)$ the all-zero vector in $\F_2^n$, by
$\F_2^{n*}$ the set $\F_2^n\setminus\{\0_n\}$, by $[n]$ the set $\{1,2,\ldots,n\}$ and by $\mathcal B_n$ the $\F_2$-vector space of $n$-variable Boolean functions $f:\F_2^n\mapsto \F_2$.
Given two vectors $x=(x_1, x_2, \ldots, x_n),y=(y_1, y_2, \ldots, y_n)$ in $\F_2^n$, we write  $x\preceq y$ and we say that $y$ {\em covers} $x$, if
$x_i\le y_i$  for all $i\in [n]$, that is, if ${\rm supp}(x)\subseteq {\rm supp}(y)$, where ${\rm supp}(x)$ is the support of $x$: ${\rm supp}(x)=\{i \in [n] : x_i\not=0 \}$.
For a vector $a=(a_1, a_2, \ldots, a_n)\in \F_2^n$, its Hamming weight $w_H(a)$
is defined as the cardinality of its support: $|{\rm supp}(a)|$.
In the rest of this paper, we use $+$ (resp. $\sum$) to denote an addition (resp. a multiple sum) considered in characteristic 0 (in $\mathbb{Z}$) and
$\oplus$ (resp. $\bigoplus$) to denote an addition (resp. a multiple sum) computed in characteristic 2 (in $\F_2$ or $\F_2^n$ or $\mathcal B_n$).\\
Given two positive integers $s,k$, we call $(s,k)$-vectorial Boolean function any function $F:\F_2^s\mapsto \F_2^k$. If $k=1$, we simply call $F$ an $s$-variable Boolean function and we define its support (resp. its co-support) as the set $\{x\in \F_2^s :\; F(x)=1\}$ (resp. $\{x\in \F_2^s :\; F(x)=0\}$).
An $n$-variable Boolean function $f$ is called bent (see \cite{Book,CMdec,Mesnager-Book,CC-Rothaus}) if its nonlinearity $2^{n-1}-\frac 12\max_{a\in \F_2^n}|W_f(a)|$ achieves the optimum $2^{n-1}-2^{\frac n2-1}$. It is called plateaued if its Walsh transform $W_f(a)=\sum_{x\in \F_2^n}(-1)^{f(x)\oplus a\cdot x}$ takes all its value in $\{0,\pm \lambda\}$ for some positive integer $\lambda$ called the amplitude (and necessarily equal to a power of 2 whose exponent is between $\frac n2$ and $n$). It is called $t$-resilient for some $t\in \{0,\dots ,n-1\}$ if its Walsh transform takes value 0 for every $a\in \F_2^n$ of Hamming weight between 0 and $t$.

\section{A generalized indirect sum construction}\label{sec: GIndirectSum}
\begin{construction}\label{D:GIndirectSum}
Let $k,s,t$ be three positive integers and  $F$ be an $(s,k)$-vectorial Boolean function.
Let ${g}$ be a Boolean function in $s$ variables and for every $u\in\F_2^k$, let $h_u$ be a Boolean functions in $t$ variables.
We construct a Boolean function $f(x,y)$ in $n=s+t$ variables as follows
\begin{eqnarray*}
f(x,y)={g}(x)\oplus h_{F(x)}(y),
\end{eqnarray*}
where $x\in\F_2^s, y\in\F_2^t$.
\end{construction}

{\color{black}\begin{remark}\label{r0}
If we restrict ourselves to the case where $F$ is a permutation, we can change the definition of $f(x,y)$ into a simpler one: $f(x,y)=h'_{F(x)}(y)$ (that is, we can without loss of generality avoid the addition of $g$). Indeed, we can take $h'_z(y)=g(F^{-1}(z))+h_z(y)$. But in our definition, $F$ is not necessarily a permutation (and its domain and co-domain are not necessarily equal), which gives more freedom and in particular allows to consider cases where the image set of $F$ is small.
\end{remark}}


\begin{remark}\label{r1}
If we compose $F$ on the left by a function whose restriction to the image set of $F$ is injective (and if we change the $h_u$'s accordingly), we do not change function $f$. If the image set of $F$ has a size smaller than or equal to some power of 2, say $2^l$, then we can without loss of generality assume that $F$ is an $(s,l)$-function.
\end{remark}

\begin{remark}\label{R:GIndirectSum}
Let $F$  be a constant function in Construction~\ref{D:GIndirectSum}, that is $|Im(F)|=|\{F(x) : x \in \F_2^s\}|=1$,
then Construction~\ref{D:GIndirectSum} becomes the direct sum construction.
\\Let $F$  be a two-valued function, that is, $Im(F)=\{{u_0},{u_1}\}\subseteq \F_2^k$.
{\color{black}Then, denoting by $1_S(x)$ the indicator (or characteristic) function of a set $S$, i.e., $1_S(x)=1$ if $x\in S$ and $1_S(x)=0$ otherwise, and denoting ${g}'=g\oplus 1_{F^{-1}({u_1})}$,  Construction~\ref{D:GIndirectSum} becomes 
 $f(x,y)={g}(x)\oplus h_{u_0}(y)\oplus ({g}\oplus {g}')(x)(h_{u_0}\oplus h_{u_1})(y)$,}
which is the indirect sum construction.\\
{\color{black}If $Im(F)=\{u_0,u_1,u_2\}$ with $u_0, u_1,u_2$ distinct,  then Construction~\ref{D:GIndirectSum}
gives: $f(x,y)=1_{F^{-1}(u_0)}(x)\big({g}(x)\oplus h_{u_0}(y)\big)\oplus 1_{F^{-1}(u_1)}(x)\big({g}(x)\oplus h_{u_1}(y)\big)\oplus 1_{F^{-1}(u_2)}(x)\big({g}(x)\oplus h_{u_2}(y)\big)$. Still denoting by $g\oplus {g}'\in\mathcal{B}_s$ the Boolean function $1_{F^{-1}(u_1)}$ whose support equals $F^{-1}(u_1)$ and by $g\oplus {g}^{''}$ the Boolean function $1_{F^{-1}(u_2)}$ whose support equals $F^{-1}(u_2)$, we have $1_{F^{-1}(u_0)}=1\oplus g'\oplus g^{''}$ and  we have then $f(x,y)=(1\oplus g'(x)\oplus g^{''}(x))\big({g}(x)\oplus h_{u_0}(y)\big)\oplus (g(x)\oplus g'(x))\big({g}(x)\oplus h_{u_1}(y)\big)\oplus (g(x)\oplus g^{''}(x))\big({g}(x)\oplus h_{u_2}(y)\big)$, that is:
$$f(x,y)=$$\begin{equation}\label{size3}{g}(x)\oplus h_{u_0}(y)\oplus (g(x)\oplus g'(x))(h_{u_1}(y)\oplus h_{u_2}(y))\oplus(g'(x)\oplus g^{''}(x))(h_{u_2}(y)\oplus h_{u_0}(y)).\end{equation} It seems impossible to obtain (\ref{gen1}) by calling $h,h',h''$ the functions $h_{u_0},h_{u_1},h_{u_2}$ in some order and the constructions seem then apparently different. Note that (\ref{gen1}) is invariant when we exchange $g$ and $h$, $g'$ and $h'$, $g''$ and $h''$, while (\ref{size3}) is not when we exchange $g$ and $h_{u_1}$, $g'$ and $h_{u_2}$, $g''$ and $h_{u_0}$.\\}
In the next theorem, we extend the study of the Walsh transforms of the direct and indirect sums to the general case of Construction~\ref{D:GIndirectSum}.
\end{remark}


{\color{black}\begin{remark}For every Boolean function $f$ over $\F_2^s\times \F_2^t$ and every Boolean function $g$ over $\F_2^s$, there exist Boolean functions $h_u$ over $\F_2^t$ and an $(s,s)$-function $F$ such that $f$ is obtained by Construction \ref{D:GIndirectSum}. Indeed, we can take, for instance, for every $x\in \F_2^s$, $F(x)=x$ and $h_x(y)=g(x)+f(x,y)$. This of course generalizes to any permutation $F$ of $\F_2^s$. We are then here generalizing a construction that has been considered in \cite{CC-C-constr} and reported in \cite[Theorem 15]{Book}, in the particular framework of bent functions. {\color{black}What is interesting with Construction \ref{D:GIndirectSum} is that we can consider cases where the image set of $F$ is small.} \end{remark}}
\subsection{The Walsh transform of the constructed functions}
In this subsection, we give two expressions of the Walsh transform of $f$, which will both be useful in the sequel.
\begin{theorem}\label{T:Walshf}
Let $k,s,t$ be three positive integers and  $F=(f_1,f_2,\cdots, f_k)$ be a $(s,k)$-vectorial Boolean function.
Let ${g}$ be a Boolean function in $s$ variables and for every $u\in\F_2^k$, let $h_u$ be a Boolean functions in $t$ variables,
and let $f$ be the $n=s+t$-variable Boolean function defined in Construction~\ref{D:GIndirectSum}. {\color{black}Then we have:
\begin{eqnarray}\label{E:fConsk=2}
f(x,y)&=&{g}(x)\oplus \bigoplus_{u\in \F_2^k}\bigoplus_{I\subseteq \{1,\dots ,k\}}\prod_{i\in I}f_i(x)\prod_{i\in \{1,\dots ,k\}\setminus I} (u_i\oplus 1)h_u(y).
\end{eqnarray}}
For any $a=(a', a^{''})\in\F_2^s\times \F_2^t$ with $a'=(a_1,\cdots,a_s)$ and $a^{''}=(a^{''}_1,\cdots, a^{''}_t)$ we have
{\color{black}\begin{eqnarray}\label{E:WalshPreimage}
 W_f(a)&=&\sum_{u\in \F_2^k}\sum_{(x,y)\in\F_2^s\times\F_2^t}1_{F^{-1}(u)}(x)(-1)^{g(x)\oplus h_u(y)\oplus a'\cdot x \oplus a^{''}\cdot y}\\\nonumber &=&\sum_{u\in \F_2^k}\sum_{x\in\F_2^s}1_{F^{-1}(u)}(x)(-1)^{g(x)\oplus a'\cdot x}W_{h_u}(a'')
\end{eqnarray}}
and
 \begin{eqnarray}\label{E:Walshfi}
 W_f(a)=\frac{1}{2^k}\sum_{u\in  \F_2^k}\sum_{v\in\F_2^k}(-1)^{v\cdot u} W_{{g}\oplus v\cdot F}(a')W_{h_u}(a^{''}),
\end{eqnarray}
where $u=(u_1,u_2,\cdots,u_k), v=(v_1,v_2,\cdots,v_k)\in\F_2^k$ and $v\cdot F=\bigoplus_{i=1}^k v_if_i$.\\
{\color{black}Defining $h(y,u)=h_u(y)$, we have:
 \begin{eqnarray}\label{E:fConsk=2Walsh}
 W_f(a)&=&\frac{1}{2^k}\sum_{v\in\F_2^k}W_{{g}\oplus v\cdot F}(a^\prime)W_h(a^{\prime\prime},v).
\end{eqnarray}}
\end{theorem}
\begin{proof}{\color{black}
Relation (\ref{E:fConsk=2}) is directly deduced from the definition of Construction \ref{D:GIndirectSum}, translated into the expression $f(x,y)={g}(x)\oplus \bigoplus_{u\in \F_2^k}\prod_{i=1}^k(f_i(x)\oplus u_i\oplus 1)h_u(y)$. \\Relation (\ref{E:WalshPreimage}) directly comes from the definition of the Walsh transform through a partition by the pre-image sets. Let us prove Relation (\ref{E:Walshfi}).
Using that, for any $c\in \F_2$, and for $1\leq  i \leq k$ and  $x\in\F_2^s$,  we have
\begin{eqnarray*}
\frac{1+(-1)^{f_i(x)\oplus c}}{2}=\left\{
             \begin{array}{ll}
               1, & \hbox{if~$f_i(x)\oplus c=0$} \\
               0 , & \hbox{if~$f_i(x)\oplus c=1$}
             \end{array}
           \right.,
\end{eqnarray*}
we can easily verify that} for any $u=(u_1,u_2,\cdots,u_k)\in\F_2^k$:
\begin{eqnarray*}
1_{F^{-1}(u)}(x)=\prod_{i=1}^{k}\left(\frac{1+(-1)^{f_i(x)\oplus u_i}}{2}\right).
\end{eqnarray*}
Therefore,  we have
\begin{eqnarray*}
W_f(a)&=&\sum_{(x,y)\in \F_2^s\times \F_2^t}(-1)^{f(x,y)\oplus a'\cdot x \oplus a^{''}\cdot y}\\
&=&\sum_{u\in \{F(z) \,:\, z \in \F_2^s\}}\sum_{(x,y)\in F^{-1}(u)\times \F_2^t}(-1)^{{g}(x)\oplus h_u(y)\oplus a'\cdot x \oplus a^{''}\cdot y}\\
&=&\sum_{u\in \{F(z) \,:\, z \in \F_2^s\}}\sum_{(x,y)\in\F_2^s\times\F_2^t}1_{F^{-1}(u)}(x)(-1)^{{g}(x)\oplus h_u(y)\oplus a'\cdot x \oplus a^{''}\cdot y}\\
&=&\sum_{u \in \F_2^k}\sum_{(x,y)\in\F_2^s\times\F_2^t}1_{F^{-1}(u)}(x)(-1)^{{g}(x)\oplus h_u(y)\oplus a'\cdot x \oplus a^{''}\cdot y}\\
&=&\sum_{u\in  \F_2^k}\sum_{(x,y)\in\F_2^s\times\F_2^t}\left[(-1)^{{g}(x)\oplus h_u(y)\oplus a'\cdot x \oplus a^{''}\cdot y}\prod_{i=1}^{k}\left(\frac{1+(-1)^{f_i(x)\oplus u_i}}{2}\right)\right]\\
&=&\frac{1}{2^k}\sum_{u\in  \F_2^k}\sum_{(x,y)\in\F_2^s\times\F_2^t}\left[(-1)^{{g}(x)\oplus h_u(y)\oplus a'\cdot x \oplus a^{''}\cdot y}\sum_{v\in\F_2^k}(-1)^{v\cdot F(x)\oplus v\cdot u}\right]\\
&=&\frac{1}{2^k}\sum_{u\in  \F_2^k}\sum_{v\in\F_2^k}(-1)^{v\cdot u} W_{{g}\oplus v\cdot F}(a')W_{h_u}(a^{''}).
\end{eqnarray*}
{\color{black}Finally, for any  $a=(a', a^{\prime\prime})\in\F_2^s\times \F_2^t$ with $a'=(a_1,\cdots,a_s)$ and $a^{\prime\prime}=(a^{\prime\prime}_1,\cdots, a^{\prime\prime}_t)$,
it follows from \eqref{E:Walshfi} that
 $$
 W_f(a)=\frac{1}{2^k}\sum_{v\in\F_2^k}W_{{g}\oplus v\cdot F}(a^\prime)\sum_{u\in  \F_2^k}(-1)^{v\cdot u} W_{h_u}(a^{\prime\prime}).$$
 {\color{black} This proves (\ref{E:fConsk=2Walsh}) since,} for any $(a^{\prime\prime},v)\in\F_2^t\times \F_2^k$, we have:
 \begin{eqnarray*}
 W_h(a^{\prime\prime},v)=&&\sum_{(y,u)\in \F_2^t\times \F_2^k}(-1)^{h(y,u)\oplus (a^{\prime\prime},v)\cdot (y,u)}\\
 &=&\sum_{u\in  \F_2^k}(-1)^{v\cdot u} W_{h_u}(a^{\prime\prime}).
 \end{eqnarray*}}
\end{proof}

{\color{black}\begin{remark} We have seen in Remark \ref{r1} that we can without loss of generality change the co-domain of $F$ (concretely, reduce the value of $k$). Note that this is true from the viewpoint of the construction of $f$ and of Relation (\ref{E:WalshPreimage}), but not really from that of Relation (\ref{E:Walshfi}). {\color{black}However,  for every $u$ outside the image set of $F$ and every $a''$, the value chosen for $W_{h_u}(a^{''})$ plays in fact no role}.
\end{remark}}

\subsection{The expression of $f$ and Relations (\ref{E:WalshPreimage}) and (\ref{E:Walshfi}) when $F$ has a small image set}
\subsubsection{$Im (F)$ of sizes 1, 2 and 3}Let us first revisit the cases of direct and indirect sums {\color{black}and study further whether Function (\ref{gen1}) recalled above can correspond to Construction \ref{D:GIndirectSum} with $|Im (F)|=3$. We shall see how Relation (\ref{E:WalshPreimage})  leads in all the first two cases to the values of $W_f(a)$ that we recalled in the introduction, and compare what it gives when $|Im (F)|=3$ with the formula we have recalled for the Walsh transform of Function (\ref{gen1})}. We shall also do the same with Relation (\ref{E:Walshfi}), since it is interesting for the study of future cases, to detail how each of these two relations fits with the known formulas. {\color{black}This will lead us (in Lemma \ref{l1}) to an observation that we consider interesting for its own sake and that will simplify the formulas that we obtain with Relations (\ref{E:WalshPreimage}) and (\ref{E:Walshfi}) and make them match.}\\
$\bullet$ {\bf Size 1:}  If $Im(F)=\{u_0\}$, the Boolean function $f$ defined in Construction~\ref{D:GIndirectSum} belongs to
the direct sum construction. Note that $1_{F^{-1}(u)}=\F_2^s$ in this case.
Relation \eqref{E:WalshPreimage} writes:
\begin{eqnarray*}
W_f(a)&=&\sum_{u\in \F_2^s}\sum_{(x,y)\in\F_2^s\times\F_2^t}1_{F^{-1}(u)}(x)(-1)^{{g}(x)\oplus h_u(y)\oplus a'\cdot x \oplus a^{''}\cdot y}\\
&=&\sum_{(x,y)\in\F_2^s\times\F_2^t}(-1)^{{g}(x)\oplus h_{u_0}(y)\oplus a'\cdot x \oplus a^{''}\cdot y}\\
&=&W_{{g}}(a')W_{h_{u_0}}(a^{''}),
\end{eqnarray*}
where $a=(a', a^{''})\in\F_2^s\times \F_2^t$ with $a'=(a_1,\cdots,a_s)$ and $a^{''}=(a^{''}_1,\cdots, a^{''}_t)$.  This is the relation recalled above for the direct sum.\\
Relation (\ref{E:Walshfi}) writes (using that $W_{g\oplus 1}=-W_g$) \begin{eqnarray*}W_f(a)&=&\frac{1}{2^k}\sum_{u\in  \F_2^k}\sum_{v\in\F_2^k}(-1)^{v\cdot u} W_{{g}\oplus v\cdot u_0}(a')W_{h_u}(a^{''})\\&=&\frac{1}{2^k}\sum_{u\in  \F_2^k}\sum_{v\in\F_2^k}(-1)^{v\cdot (u+u_0)} W_{{g}}(a')W_{h_u}(a^{''})\\&=&W_{{g}}(a')W_{h_{u_0}}(a^{''}).\end{eqnarray*}\\
$\bullet$ {\bf Size 2:} If $Im(F)=\{u_0,u_1\}$ with $u_0\neq u_1$,  then as we also saw already,  Construction~\ref{D:GIndirectSum}
becomes the indirect sum. {\color{black}As we observed in Remark \ref{r1}, we could\footnote{This is also true for the previous case.}  take $F$ valued in $\F_2$, and so $k=1$ (fixing for instance $u_0=0$ and $u_1=1$). We prefer however not to make such a simplification, for two reasons: this will better prepare us for the study of the case where $Im (F)$ has size 3 below (that we shall also make without restricting the co-domain of $F$), and this keeps full freedom for the choice of function $F$.} We have: $f(x,y)=1_{F^{-1}(u_0)}(x)\big({g}(x)\oplus h_{u_0}(y)\big)\oplus 1_{F^{-1}(u_1)}(x)\big({g}(x)\oplus h_{u_1}(y)\big)$, and for any $a=(a', a^{''})\in\F_2^s\times \F_2^t$, where  $a'=(a_1,\cdots,a_s)$ and $a^{''}=(a^{''}_1,\cdots, a^{''}_t)$,
Relation (\ref{E:WalshPreimage}) gives:
\begin{eqnarray*}
W_f(a)&=&\sum_{u\in \F_2^s}\sum_{(x,y)\in\F_2^s\times\F_2^t}1_{F^{-1}(u)}(x)(-1)^{{g}(x)\oplus h_u(y)\oplus a'\cdot x \oplus a^{''}\cdot y}\\
&=&\sum_{(x,y)\in\F_2^s\times\F_2^t}1_{F^{-1}(u_0)}(x)(-1)^{{g}(x)\oplus h_{u_0}(y)\oplus a'\cdot x \oplus a^{''}\cdot y}\\
&&+\sum_{(x,y)\in\F_2^s\times\F_2^t}1_{F^{-1}(u_1)}(x)(-1)^{{g}(x)\oplus h_{u_1}(y)\oplus a'\cdot x \oplus a^{''}\cdot y}.
\end{eqnarray*}
Still denoting by $g\oplus {g}'\in\mathcal{B}_s$ the Boolean function $1_{F^{-1}(u_1)}$ whose support equals $F^{-1}(u_1)$, we have that the co-support of $g\oplus {g}'\in\mathcal{B}_s$ is $F^{-1}(u_0)$ (and we have seen in Remark \ref{R:GIndirectSum}  that $f(x,y)=
{g}(x)\oplus h_{u_0}(y)\oplus ({g}\oplus {g}')(x)(h_{u_0}\oplus h_{u_1})(y)$) and using that $1\oplus (g\oplus {g}')(x)=\frac 12{\Big (}1+(-1)^{(g\oplus {g}')(x)}{\Big )}$ and $(g\oplus {g}')(x)=\frac 12{\Big (}1-(-1)^{(g\oplus {g}')(x)}{\Big )}$, we have from the latter expression we obtained for $W_f(a)$ above:
\begin{eqnarray}\label{Eq:IndirectsumWalsh}
W_f(a', a^{''})=\frac 12 W_{{g}}(a')\big[W_{h_{u_0}}(a^{''})+W_{h_{u_1}}(a^{''})\big]+\frac 12 W_{{g}'}(a')\big[W_{h_{u_0}}(a^{''})-W_{h_{u_1}}(a^{''})\big],
\end{eqnarray}and this is the relation that we recalled in the introduction. Let  us now consider Relation (\ref{E:Walshfi}).
For every $u\in  \F_2^k$, we have  (using again that $W_{g\oplus 1}=-W_g$):
$$\sum_{v\in\F_2^k}(-1)^{v\cdot u} W_{{g}\oplus v\cdot F}(a')=$$$$\sum_{v\in\F_2^k}(-1)^{v\cdot u} W_{{g}\oplus (v\cdot u_0)(g\oplus {g}'\oplus 1)\oplus (v\cdot u_1)(g\oplus {g}')}(a')=$$$$\sum_{v\in\F_2^k}(-1)^{v\cdot (u\oplus u_0)} W_{{g}\oplus (v\cdot (u_0\oplus u_1))(g\oplus {g}')}(a')=$$$$\sum_{v\in \langle u_0\oplus u_1\rangle^\perp}(-1)^{v\cdot (u\oplus u_0)} W_{g}(a')+\sum_{v\in \F_2^k\setminus \langle u_0\oplus u_1\rangle^\perp}(-1)^{v\cdot (u\oplus u_0)} W_{g'}(a')=$$$$\left\{\begin{array}{l}2^{k-1}W_g(a')\mbox{ if }u\in \{u_0,u_1\}\\0\mbox{ otherwise}\end{array}\right.+\left\{\begin{array}{l}2^{k-1}W_g(a')\mbox{ if }u=u_0\\-2^{k-1}W_g(a')\mbox{ if }u=u_1\\0\mbox{ otherwise}\end{array}\right..$$
Relation (\ref{E:Walshfi}) gives then (\ref{Eq:IndirectsumWalsh}).

{\color{black}$\bullet$ {\bf Size 3:} If $Im(F)=\{u_0,u_1,u_2\}$ with $u_0, u_1,u_2$ distinct,  then we have seen in Remark \ref{R:GIndirectSum} that Construction~\ref{D:GIndirectSum}
gives: $f(x,y)=1_{F^{-1}(u_0)}(x)\big({g}(x)\oplus h_{u_0}(y)\big)\oplus 1_{F^{-1}(u_1)}(x)\big({g}(x)\oplus h_{u_1}(y)\big)\oplus 1_{F^{-1}(u_2)}(x)\big({g}(x)\oplus h_{u_2}(y)\big)$, and for any $a=(a', a^{''})\in\F_2^s\times \F_2^t$, where  $a'=(a_1,\cdots,a_s)$ and $a^{''}=(a^{''}_1,\cdots, a^{''}_t)$,
Relation (\ref{E:WalshPreimage}) gives:
\begin{eqnarray}
\nonumber W_f(a)&=&\sum_{u\in \F_2^s}\sum_{(x,y)\in\F_2^s\times\F_2^t}1_{F^{-1}(u)}(x)(-1)^{{g}(x)\oplus h_u(y)\oplus a'\cdot x \oplus a^{''}\cdot y}\\\nonumber
&=&\sum_{(x,y)\in\F_2^s\times\F_2^t}1_{F^{-1}(u_0)}(x)(-1)^{{g}(x)\oplus h_{u_0}(y)\oplus a'\cdot x \oplus a^{''}\cdot y}\\\nonumber
&&+\sum_{(x,y)\in\F_2^s\times\F_2^t}1_{F^{-1}(u_1)}(x)(-1)^{{g}(x)\oplus h_{u_1}(y)\oplus a'\cdot x \oplus a^{''}\cdot y}\\\label{eqqq}
&&+\sum_{(x,y)\in\F_2^s\times\F_2^t}1_{F^{-1}(u_2)}(x)(-1)^{{g}(x)\oplus h_{u_2}(y)\oplus a'\cdot x \oplus a^{''}\cdot y}.
\end{eqnarray}
Denoting as in Remark \ref{R:GIndirectSum} by $g\oplus {g}'\in\mathcal{B}_s$ the Boolean function $1_{F^{-1}(u_1)}$ and by $g\oplus {g}^{''}$ the Boolean function $1_{F^{-1}(u_2)}$ (which gives $f(x,y)={g}(x)\oplus h_{u_0}(y)\oplus (g(x)\oplus g'(x))(h_{u_1}(y)\oplus h_{u_2}(y))\oplus(g'(x)\oplus g^{''}(x))(h_{u_2}(y)\oplus h_{u_0}(y))$) and using that $1_{F^{-1}(u_0)}=1\oplus (g'\oplus {g}^{''})(x)=\frac 12{\Big (}1+(-1)^{(g'\oplus {g}^{''})(x)}{\Big )}$ and $1_{F^{-1}(u_1)}=(g\oplus {g}')(x)=\frac 12{\Big (}1-(-1)^{(g\oplus {g}')(x)}{\Big )}$ and $1_{F^{-1}(u_2)}=(g\oplus {g}^{''})(x)=\frac 12{\Big (}1-(-1)^{(g\oplus {g}^{''})(x)}{\Big )}$, we deduce from Relation (\ref{eqqq}):
\begin{eqnarray}
\nonumber W_f(a)&=&\frac 12\sum_{(x,y)\in\F_2^s\times\F_2^t}{\big (}(-1)^{{g}(x)\oplus h_{u_0}(y)\oplus a'\cdot x \oplus a^{''}\cdot y}+(-1)^{{g}(x)\oplus g' (x)\oplus g^{''}(x)\oplus h_{u_0}(y)\oplus a'\cdot x \oplus a^{''}\cdot y}{\Big )}\\
\nonumber &&+\frac 12\sum_{(x,y)\in\F_2^s\times\F_2^t}{\big (}(-1)^{{g}(x)\oplus h_{u_1}(y)\oplus a'\cdot x \oplus a^{''}\cdot y}-(-1)^{g' (x)\oplus h_{u_1}(y)\oplus a'\cdot x \oplus a^{''}\cdot y}{\Big )}\\
\nonumber &&+\frac 12\sum_{(x,y)\in\F_2^s\times\F_2^t}{\big (}(-1)^{{g}(x)\oplus h_{u_2}(y)\oplus a'\cdot x \oplus a^{''}\cdot y}-(-1)^{ g^{''}(x)\oplus h_{u_2}(y)\oplus a'\cdot x \oplus a^{''}\cdot y}{\Big )}\\\nonumber &=&\frac 12 W_{h_{u_0}}(a^{''}){\big (}W_g(a')+W_{g\oplus g'\oplus g^{''}}(a'){\Big )}+\frac 12W_{h_{u_1}}(a^{''}){\big (}W_g(a')-W_{g'}(a'){\Big )}\\
\label{walshgen1}&&+\frac 12W_{h_{u_2}}(a^{''}){\big (}W_g(a')-W_{g^{''}}(a'){\Big )}.
\end{eqnarray}
This is different from what we recalled in the introduction about Construction \ref{gen1}'s Walsh transform and this confirms that we have here a different construction. Note that the expression of the  Walsh transform is simpler here, but we are in a particular case since we have $(g\oplus g')(g\oplus g'')=0$ while there is no condition on $g,g',g''$ in (\ref{gen1}).\\
Let us now study what gives Relation (\ref{E:Walshfi}). For every $u\in  \F_2^k$, we have  (using again that $W_{g\oplus 1}=-W_g$):
$$\sum_{v\in\F_2^k}(-1)^{v\cdot u} W_{{g}\oplus v\cdot F}(a')=$$$$\sum_{v\in\F_2^k}(-1)^{v\cdot u} W_{{g}\oplus (v\cdot u_0)(1\oplus g'\oplus g^{''})\oplus (v\cdot u_1)(g\oplus g'))\oplus (v\cdot u_2)(g\oplus g^{''})}(a')=$$$$\sum_{v\in\F_2^k}(-1)^{v\cdot (u\oplus u_0)} W_{{g}\oplus (v\cdot (u_0\oplus u_1))(g\oplus {g}')\oplus (v\cdot (u_0\oplus u_2))(g\oplus {g}'')}(a')=$$$$\sum_{v\in \langle u_0\oplus u_1,u_0\oplus u_2\rangle^\perp}(-1)^{v\cdot (u\oplus u_0)} W_{g}(a')+\sum_{v\in \langle u_0\oplus u_1\rangle^\perp\setminus \langle u_0\oplus u_2\rangle^\perp}(-1)^{v\cdot (u\oplus u_0)} W_{g''}(a')+$$$$\sum_{v\in \langle u_0\oplus u_2\rangle^\perp\setminus \langle u_0\oplus u_1\rangle^\perp}(-1)^{v\cdot (u\oplus u_0)} W_{g'}(a')+\sum_{v\in \F_2^k\setminus (\langle u_0\oplus u_1\rangle^\perp \cup \langle u_0\oplus u_2\rangle^\perp)}(-1)^{v\cdot (u\oplus u_0)} W_{g\oplus g'\oplus g''}(a')=$$$$\left\{\begin{array}{l}2^{k-2}W_g(a')\mbox{ if }u\in \{u_0,u_1,u_2,u_0\oplus u_1\oplus u_2\}\\0\mbox{ otherwise}\end{array}\right.+\left\{\begin{array}{l}2^{k-2}W_{g''}(a')\mbox{ if }u\in \{u_0,u_1\}\\-2^{k-2}W_{g''}(a')\mbox{ if }u\in \{u_2,u_0\oplus u_1\oplus u_2\}\\0\mbox{ otherwise}\end{array}\right.$$$$+\left\{\begin{array}{l}2^{k-2}W_{g'}(a')\mbox{ if }u\in \{u_0,u_2\}\\-2^{k-2}W_{g'}(a')\mbox{ if }u\in \{u_1,u_0\oplus u_1\oplus u_2\}\\0\mbox{ otherwise}\end{array}\right.+\left\{\begin{array}{l}2^{k-2}W_{g\oplus g'\oplus g''}(a')\mbox{ if }u\in \{u_0,u_0\oplus u_1\oplus u_2\}\\-2^{k-2}W_{g\oplus g'\oplus g''}(a')\mbox{ if }u\in \{u_1,u_2\}\\0\mbox{ otherwise}\end{array}\right..$$
Relation (\ref{E:Walshfi}) gives then $$W_f(a)=\frac 14W_g(a'){\Big (}W_{h_{u_0}}(a^{''})+W_{h_{u_1}}(a^{''})+W_{h_{u_2}}(a^{''})+W_{h_{u_0\oplus u_1\oplus u_2}}(a^{''}){\Big )}+$$$$\frac 14W_{g''}(a'){\Big (}W_{h_{u_0}}(a^{''})+W_{h_{u_1}}(a^{''})-W_{h_{u_2}}(a^{''})-W_{h_{u_0\oplus u_1\oplus u_2}}(a^{''}){\Big )}+$$$$\frac 14W_{g'}(a'){\Big (}W_{h_{u_0}}(a^{''})-W_{h_{u_1}}(a^{''})+W_{h_{u_2}}(a^{''})-W_{h_{u_0\oplus u_1\oplus u_2}}(a^{''}){\Big )}+$$\begin{equation}\label{exp}\frac 14W_{g\oplus g'\oplus g''}(a'){\Big (}W_{h_{u_0}}(a^{''})-W_{h_{u_1}}(a^{''})-W_{h_{u_2}}(a^{''})+W_{h_{u_0\oplus u_1\oplus u_2}}(a^{''}){\Big )},\end{equation} which is different (as an expression) from (\ref{walshgen1}) but necessarily has the same value. Note in particular that, if we develop the expression, the term $W_{h_{u_0\oplus u_1\oplus u_2}}(a^{''})$ has for factor $W_g(a')-W_{g''}(a')-W_{g'}(a')+W_{g\oplus g'\oplus g''}(a')$. We can check that this factor equals 0 (as it should, because $u_0\oplus u_1\oplus u_2$ being different from $u_0$, $u_1$ and $u_2$, the value of $W_{h_{u_0\oplus u_1\oplus u_2}}(a^{''})$ is supposed to play no role in the value of $W_f$) because the two functions $g\oplus g'$ and $g\oplus g''$ have disjoint supports. We check this in the next lemma, which will alow to moreover simplify Relation (\ref{exp}). Note that Expression (\ref{exp}) is similar but different from (\ref{walshant}).
\begin{lemma}\label{l1}Let $g,g',g''$ be any Boolean functions such that $(g\oplus g')(g\oplus g'')=0$. We have then that the integer-valued functions $g-g'-g''+(g\oplus g'\oplus g'')$ and $W_g-W_{g'}-W_{g''}+W_{g\oplus g'\oplus g''}$ are identically zero.
\end{lemma}
\begin{proof}We have $g\oplus g'=g+g'-2gg'$ and $g\oplus g''=g+g''-2gg''$. Then $0=(g+g'-2gg')(g+g''-2gg'')=g-gg'-gg''+g'g''$; hence, $g'g''=gg'+gg''-g$. We have $g\oplus g'\oplus g''=g+g'+g''-2gg'-2gg''-2g'g''+4gg'g''=g+g'+g''-2gg'-2gg''-2g'g''+4(gg'+gg''-g)$ and then $g-g'-g''+g\oplus g'\oplus g''=2gg'+2gg''-2g'g''-2g=0$. For every $s$-variable Boolean function $h$, we have $W_h(a')=2^s\delta_0(a')-2\widehat h(a')$ where $\widehat h(a')=\sum_{x\in \F_2^s}h(x)(-1)^{a'\cdot x}$ and we easily deduce from the relation $g-g'-g''+g\oplus g'\oplus g''=0$ and from the linearity of the Fourier transform with respect to the addition in $\mathbb Z$ that $W_g(a')-W_{g''}(a')-W_{g'}(a')+W_{g\oplus g'\oplus g''}(a')=0$.\end{proof}
\begin{remark}Lemma \ref{l1} can be interpreted in the following way (after denoting $f=g\oplus g'$ and $h=g\oplus g''$) which may present some interest in the framework of algebraic attacks: let $f$ be any Boolean function and $h$ be any annihilator of $f$ (i.e.satisfying $fg= 0$). Let $g$ be any Boolean function, then we have $W_g-W_{g\oplus f}-W_{g\oplus h}+W_{g\oplus f\oplus h}= 0$.
\end{remark}
Using now Lemma \ref{l1}, we obtain from Relation (\ref{exp}): $W_f(a)=\frac 14W_g(a'){\Big (}W_{h_{u_0}}(a^{''})+W_{h_{u_1}}(a^{''})+W_{h_{u_2}}(a^{''})+W_{h_{u_0\oplus u_1\oplus u_2}}(a^{''}){\Big )}+\frac 14W_{g''}(a'){\Big (}W_{h_{u_0}}(a^{''})+W_{h_{u_1}}(a^{''})-W_{h_{u_2}}(a^{''})-W_{h_{u_0\oplus u_1\oplus u_2}}(a^{''}){\Big )}+\frac 14W_{g'}(a'){\Big (}W_{h_{u_0}}(a^{''})-W_{h_{u_1}}(a^{''})+W_{h_{u_2}}(a^{''})-W_{h_{u_0\oplus u_1\oplus u_2}}(a^{''}){\Big )}+\frac 14{\big (}W_{g'}(a')+W_{g''}(a')-W_{g}(a'){\big )}{\Big (}W_{h_{u_0}}(a^{''})-W_{h_{u_1}}(a^{''})-W_{h_{u_2}}(a^{''})+W_{h_{u_0\oplus u_1\oplus u_2}}(a^{''}){\Big )}=\frac 14W_g(a'){\Big (}2W_{h_{u_1}}(a^{''})+2W_{h_{u_2}}(a^{''}){\Big )}+\frac 14W_{g''}(a'){\Big (}2W_{h_{u_0}}(a^{''})-2W_{h_{u_2}}(a^{''}){\Big )}+\frac 14W_{g'}(a'){\Big (}2W_{h_{u_0}}(a^{''})-2W_{h_{u_1}}(a^{''}){\Big )}$, and we obtain from Relation (\ref{walshgen1}):$$W_f(a)=\frac 12 W_{h_{u_0}}(a^{''}){\big (}W_{g'}(a')+W_{g^{''}}(a'){\big )}+$$ \begin{equation}\label{finaleq}\frac 12W_{h_{u_1}}(a^{''}){\big (}W_g(a')-W_{g'}(a'){\big )}+\frac 12W_{h_{u_2}}(a^{''}){\big (}W_g(a')-W_{g^{''}}(a'){\big )},\end{equation}which is the same.

\begin{remark}In Relation (\ref{walshant}), to make the coefficient of $W_{h\oplus h'\oplus h^{''}}(a^{''})$ equal to 0, we are led to take $g,g',g''$ such that $W_{g\oplus g'\oplus g^{''}}(a')+W_{g^{''}}(a')=W_{g}(a')\!+\!W_{g'}(a')\!$. For being able to apply Lemma \ref{l1}, we exchange the roles of $g$ and $g''$ and we get then $W_f(a',a^{''})= \frac{1}{2}W_{h}(a^{''})\left[ W_{g''}(a')+W_{g'}(a')\right]+\frac{1}{2}W_{h'}(a^{''})\left[ W_{g''}(a')-W_{g}(a')\right]+\frac{1}{2}W_{h^{''}}(a^{''})\left[ W_{g}(a')-W_{g'}(a')\right]$, similar to (\ref{finaleq}) but not the same.
\end{remark}}
\subsubsection{$Im (F)$ of size 4}{\color{black}In this case, we apply what we observed in Remark \ref{r1} and take then $k=2$, since continuing in the most general setting as we did with the previous cases would lead to complex relations. Let us recall Theorem \ref{T:Walshf} in the particular case of $k=2$.}
Given  $F=(f_1,f_2)$ a $(s,2)$-vectorial Boolean function with image set $\{(0,0),(0,1),(1,0),(1,1)\}$ and $g$ an $s$-variable Boolean function,  and $h_{(0,0)}, h_{(0,1)}, h_{(1,0)}, h_{(1,1)}$ four Boolean functions in $t$ variables,
we define the $n=s+t$-variable Boolean function $f(x,y)={g}(x)\oplus h_{F(x)}(y)$,
where $x\in\F_2^s, y\in\F_2^t$. Then we have
\begin{eqnarray}\label{E:fConsk=2'}
f(x,y)&=&g(x)\oplus f_1(x)f_2(x)[h_{(0,0)}(y)\oplus h_{(0,1)}(y)\oplus h_{(1,0)}(y)\oplus h_{(1,1)}(y)]\nonumber\\
&&\oplus f_1(x)[h_{(0,0)}(y)\oplus h_{(1,0)}(y)]\oplus f_2(x)[h_{(0,0)}(y)\oplus h_{(0,1)}(y)]\oplus h_{(0,0)}(y).
\end{eqnarray}
Defining $h(y,u)=h_u(y)$, that is, $h=h_{(0,0)}|| h_{(0,1)}||h_{(1,0)}||h_{(1,1)}$, we have:
 \begin{eqnarray}\label{E:fConsk=2Walsh'}
 W_f(a)&=&\frac{1}{4}\Big[W_h(a^{''},0,0)W_g(a^{'})+W_h(a^{''},0,1)W_{g\oplus f_2}(a^{'})\nonumber\\
 &&+W_h(a^{''},1,0)W_{g\oplus f_1}(a^{'})+W_h(a^{''},1,1)W_{g\oplus f_1\oplus f_2}(a^{'})\Big].
\end{eqnarray}

{\color{black}\paragraph{Function (\ref{gen1}) as a particular case:}
Let $g,g',g^{''}$ be three Boolean functions in $s$ variables such that $f_1(x)=(g\oplus g')(x)$ and $f_2(x)=(g\oplus g^{''})(x)$. Assuming that: $$h_{(0,1)}=h_{(0,0)}\oplus h_{(1,0)}\oplus  h_{(1,1)},$$ there exist three Boolean functions $h,h',h^{''}$ in $t$ variables such that} $h_{(0,0)}(y)=h(y)$, $h_{(0,1)}(y)=(h\oplus h'\oplus h^{''})(y)$,
$h_{(1,0)}(y)=h^{''}(y)$,  $h_{(1,1)}(y)=h'(y)$.
Then by~\eqref{E:fConsk=2'} we have
\begin{eqnarray*}
f(x,y)&=&g(x)\oplus h(y)\oplus (g\oplus g')(x)(h\oplus h^{''})(y)\oplus (g\oplus g^{''})(x)(h'\oplus h^{''})(y)\\
&=&g(x)\oplus h(y)\oplus (g \oplus g^{'})(x) (h \oplus h^{'})(y) \oplus (g^{'} \oplus g'')(x)(h^{'} \oplus h'')(y)
\end{eqnarray*}
which is equal to Function~\eqref{gen1} and the Walsh spectrum given by~\eqref{E:fConsk=2Walsh'} is the same as~\eqref{walshant}.
{\color{black}Indeed, ~\eqref{walshant} can be rewrote as
\begin{eqnarray}\label{walshantrewrite}
W_f(a',a^{''})&=& \frac{1}{4}W_{g}(a')\left[W_h(a^{''})+W_{h\oplus h'\oplus h^{''}}(a^{''})+W_{h^{''}}(a^{''})+W_{h'}(a^{''})\right]\nonumber\\
&&+ \frac{1}{4}W_{g^{''}}(a')\left[W_h(a^{''})-W_{h\oplus h'\oplus h^{''}}(a^{''})+W_{h^{''}}(a^{''})-W_{h'}(a^{''})\right]\nonumber\\
&&+ \frac{1}{4}W_{g'}(a')\left[W_h(a^{''})+W_{h\oplus h'\oplus h^{''}}(a^{''})-W_{h^{''}}(a^{''})-W_{h'}(a^{''})\right]\nonumber\\
&&+ \frac{1}{4}W_{g\oplus g'\oplus g^{''}}(a')\left[W_h(a^{''})-W_{h\oplus h'\oplus h^{''}}(a^{''})-W_{h^{''}}(a^{''})+W_{h'}(a^{''})\right]\nonumber\\
&=&\frac{1}{4}\Big[W_{\tilde{h}}(a^{''},0,0) W_{g}(a')+W_{\tilde{h}}(a^{''},0,1) W_{g^{''}}(a')+W_{\tilde{h}}(a^{''},1,0) W_{g'}(a')\nonumber\\
&&+W_{\tilde{h}}(a^{''},1,1) W_{g\oplus g'\oplus g^{''}}(a')\Big],
\end{eqnarray}
where $\tilde{h}=h||h\oplus h'\oplus h^{''}||h^{''}||h'$.
Recall that $f_1(x)=(g\oplus g')(x)$, $f_2(x)=(g\oplus g^{''})(x)$, $h_{(0,0)}(y)=h(y)$, $h_{(0,1)}(y)=(h\oplus h'\oplus h^{''})(y)$,
$h_{(1,0)}(y)=h^{''}(y)$,  $h_{(1,1)}(y)=h'(y)$. \eqref{walshantrewrite} is the same as~\eqref{E:fConsk=2Walsh'} by replacing $\tilde{h}$ with $h$.
}

{\color{black}\subsubsection{$Im (F)$ of size  at most 8}}
We now consider the vectorial Boolean function $F(x)=(f_1(x), f_2(x), f_3(x))\in \F_2^3$ in Construction~\ref{D:GIndirectSum}.
According to Theorem \ref{T:Walshf}, we have:
\begin{eqnarray}\label{E:Im8}
f(x,y)&=& g(x)\oplus h_{(0,0,0)}(y)\oplus f_1(x)f_2(x)f_3(x)\nonumber\\
&&[(h_{(0,0,0)}\oplus h_{(0,0,1)}\oplus h_{(0,1,0)}\oplus h_{(0,1,1)}\oplus h_{(1,0,0)}\oplus h_{(1,0,1)}\oplus h_{(1,1,0)}\oplus h_{(1,1,1)})(y)]\nonumber\\
&&\oplus f_1(x)f_2(x)[(h_{(0,0,0)}\oplus h_{(0,1,0)}\oplus h_{(1,0,0)}\oplus h_{(1,1,0)})(y)]\nonumber\\
&&\oplus f_1(x)f_3(x)[(h_{(0,0,0)}\oplus h_{(1,0,0)}\oplus h_{(0,0,1)}\oplus h_{(1,0,1)})(y)]\nonumber\\
&&\oplus f_2(x)f_3(x)[(h_{(0,0,0)}\oplus h_{(0,0,1)}\oplus h_{(0,1,0)}\oplus h_{(0,1,1)})(y)]\nonumber\\
&&\oplus f_1(x)[(h_{(0,0,0)}\oplus h_{(1,0,0)})(y)]\oplus f_2(x)[(h_{(0,0,0)}\nonumber \\&&\oplus h_{(0,1,0)})(y)]\oplus f_3(x)[(h_{(0,0,0)}\oplus h_{(0,0,1)})(y)].
\end{eqnarray}
{\color{black}\paragraph{Function~\eqref{gen2} as a particular case:}
Let $g,g', g^{''}, g^{'''}$ be four $s$-variable Boolean functions such that $f_1(x)=(g\oplus g')(x), f_2(x)=(g\oplus g^{''})(x),  f_3(x)=(g\oplus g^{'''})(x))$, and assuming:
\begin{eqnarray*}&&h_{(0,0,1)}=h_{(0,0,0)}\oplus  h_{(1,0,0)}\oplus  h_{(1,0,1)},\\
&&h_{(0,1,0)}=h_{(0,0,0)}\oplus  h_{(1,1,1)}\oplus  h_{(1,0,1)},\\
&&h_{(0,1,1)}=h_{(0,0,0)}\oplus  h_{(1,1,1)}\oplus  h_{(1,0,0)},\\
&&h_{(1,1,0)}=h_{(1,1,1)}\oplus  h_{(1,0,0)}\oplus  h_{(1,0,1)},\end{eqnarray*}
let $h, h', h^{''}, h^{'''}$ be four $t$-variable Boolean functions such that:}
$h_{(0,0,0)}(y)=h(y)$,
$h_{(0,0,1)}(y)=(h\oplus  h^{''}\oplus  h^{'''})(y)$,
$h_{(0,1,0)}(y)=(h\oplus  h^{'}\oplus  h^{'''})(y)$,
$h_{(0,1,1)}(y)=(h\oplus  h^{'}\oplus  h^{''})(y)$,
$h_{(1,0,0)}(y)=h^{''}(y)$,
$h_{(1,0,1)}(y)=h^{'''}(y)$,
$h_{(1,1,0)}(y)=(h^{'}\oplus  h^{''}\oplus  h^{'''})(y)$, and
$h_{(1,1,1)}(y)=h^{'}(y)$.
Then the coefficients of $f_1(x)f_2(x)f_3(x), f_1(x)f_2(x), f_1(x)f_3(x)$, and $f_2(x)f_3(x)$ are all equal to 0 in~\eqref{E:Im8}, which becomes
\begin{eqnarray*}
&&f(x,y)\\
&=&g(x)\oplus h(y)\oplus f_1(x)[(h_{(0,0,0)}\oplus h_{(1,0,0)})(y)]\oplus f_2(x)[(h_{(0,0,0)}\oplus h_{(0,1,0)})(y)]\oplus \\&&f_3(x)[(h_{(0,0,0)}\oplus h_{(0,0,1)})(y)]\\
&=&g(x)\oplus h(y)\oplus (g\oplus g')(x)(h\oplus h^{''})(y)\oplus \\&&(g\oplus g^{''})(x)(h'\oplus h^{'''})(y)
\oplus (g\oplus g^{'''})(x)(h^{''}\oplus h^{'''})(y)\\
&=&g(x)\oplus h(y)\oplus g(x)h(y)\oplus g'(x)h(y)\oplus g'(x)h^{''}(y)\oplus g(x)h'(y)\\
&&\oplus g^{''}(x)h'(y)\oplus g^{''}(x)h^{'''}(y)\oplus g^{'''}(x)h^{''}(y)\oplus g^{'''}(x)h^{'''}(y)
\end{eqnarray*}
which is equal to Function~\eqref{gen2}. {\color{black}We can then consider Construction~\ref{D:GIndirectSum} as a proper general construction covering as particular cases most of the main known secondary constructions. }

\end{document}